\documentclass[a4paper,12pt]{article}
\usepackage{amsmath,amsthm,amsfonts,mathtools}
\usepackage{amssymb}
\usepackage{hyperref}
\usepackage[T1]{fontenc}
\usepackage{times}
\mathtoolsset{showonlyrefs=true}
\usepackage{paralist}
\usepackage{graphicx}
\setdefaultenum{i)}{(a)}{}{}
\usepackage[round]{natbib}
\usepackage[margin=1in]{geometry}

\numberwithin{equation}{section}
\usepackage{enumerate}														
\usepackage{color}

\newtheorem{theorem}{Theorem}[section]

\newtheorem{assumption}[theorem]{Assumption}
\newtheorem{definition}[theorem]{Definition}
\newtheorem{example}[theorem]{Example}
\newtheorem{lemma}[theorem]{Lemma}
\newtheorem{proposition}[theorem]{Proposition}
\newtheorem{remark}[theorem]{Remark}

\numberwithin{equation}{section}

\newcommand{\nada}[1]{}

\DeclareMathOperator{\avtrco}{\textup{ATC}}
\DeclareMathOperator{\esr}{\textup{ESR}}
\usepackage{csquotes}

\long\def\symbolfootnote[#1]#2{\begingroup\def\thefootnote{\fnsymbol{footnote}}\footnote[#1]{#2}\endgroup}

\title{Almost Perfect Shadow Prices}
\author{Eberhard Mayerhofer\footnote{University of Limerick, Department of Mathematics \& Statistics, V94 T9PX, Limerick, Ireland, \texttt{eberhard.mayerhofer@ul.ie}.}}
\begin{document}
\maketitle

\abstract{
Shadow prices simplify the derivation of optimal trading strategies in markets with transaction costs by transferring optimization into a more tractable, frictionless market. This paper establishes that a na\"ive shadow price Ansatz for maximizing long term returns given average volatility yields a strategy that is, for small bid-ask-spreads, asymptotically optimal at third order. Considering the second-order impact of transaction costs, such a strategy is essentially optimal. However, for risk aversion different from one, we devise alternative strategies that outperform the shadow market at fourth order. Finally, it is shown that the risk-neutral objective rules out the existence of shadow prices.
}

\vspace{0.5cm}

\newpage
\section{Introduction}
\begin{quote}
With a little help of my friends.\footnote{This paper relies on MATHEMATICA (\url{https://www.wolfram.com/mathematica/}) to approximate several trading strategies and performances for small bid-ask spreads. For motivating this research topic, sharing asymptotic methods and providing feedback, I am indebted to Professor Paolo Guasoni.}.
\par\hfill--- The Beatles
\end{quote}
{\it Shadow prices} (or {\it consistent price systems}) are a tool for characterizing the absence of arbitrage opportunities in markets with proportional transaction costs (see, e.g., \cite{kabanov2002no,gr2010,czichowsky2016duality}), 
or for deriving optimal strategies for various objectives (see, e.g., \cite{kallsen2010using,gerhold2013dual,gerhold.al.11,guasonimuhlekarbe,herdegen}). This paper investigates the applicability of shadow prices to the optimisation of long term returns given average volatility.

Strategies that are optimal in frictionless markets\footnote{More generally, the term {\it market frictions} encompasses, e.g. price impact, short-selling constraints, and margin requirements (see \cite{guasonimuhlekarbe}, \cite{weber} and \cite{gmz2023} and the references therein.} such as delta-hedging of European-type options, or constant proportion strategies, lead to immediate bankruptcy under proportional costs.\footnote{Example \ref{ex: swap} below shows this failure for a variance swap hedge.} To ensure solvency, trading frequency needs to be modulated to finite variation, trading as little as necessary to stay close to the target exposures. This paper relates to objectives of long-run investors (\cite{taksar1988diffusion,gerhold.al.11,gm2020,gm2023})\footnote{These references deal with particularly tractable, long-run problems of local or global utility maximisation; however, the first papers in this field, starting with \cite{magill}, where optimal investment and consumption problems on an infinite horizon, which exhibit similar strategies
and asymptotics. For an overview of this research field, see \cite[Chapter 1]{gm2020} and \cite{guasonimuhlekarbe}.}, who consider it optimal to keep the fraction of wealth $\pi$ invested in the risky asset within an interval around a target exposure, by engaging only in trading whenever this fraction hits the boundaries of the interval.\footnote{For these strategies, the name ``control limit policy" from \cite{taksar1988diffusion} is adopted, see Definition \ref{def: controllable} below.}  For example, for constant investment opportunities and a sufficiently small relative bid-ask spread $\varepsilon$, the trading boundaries $\pi_-<\pi_+$
of an investor with risk-aversion $\gamma$ are approximately
\begin{equation}\label{eq: trafosx}
\pi_\pm=\pi_*\pm\left(\frac{3}{4\gamma}\pi_*^2(\pi_*-1)^2\right)^{1/3}\varepsilon^{1/3},
\end{equation}
where $\pi_*=\frac{\mu}{\gamma \sigma^2}$ is the well-known Merton fraction, $\mu$ being the annualized average return of the risky asset, and $\sigma$ its volatility.

Absurdly, such finite variation strategies are mathematically more challenging than the infinite variation strategies typically encountered in frictionless markets. Shadow prices allow to transfer optimization into a more tractable, frictionless (but fictitious) market. A shadow price is  a frictionless asset that evolves in the bid-ask spread of the risky asset and for which the optimal strategy buys (resp. sells) whenever its price agrees with the ask (resp. bid) of this risky asset. For objectives which are monotone functions of wealth, such as power utility, the strategy in the shadow market is also optimal in the original market, because by trading in the shadow market, the investor is generally better off. Furthermore, shadow markets provide an elegant, intuitive derivation of optimal trading policies for different objectives and market models. It therefore may come as surprise that \cite{gm2020,gm2023} use the more traditional Hamilton-Jacobi-Bellman equations both for the heuristic derivation of the candidate optimal control limit policy (with asymptotics \eqref{eq: trafosx}) {\it and} the verification of optimality. This is even more surprising, as the respective objectives lend themselves to very tractable candidate shadow prices and trading strategies (see Section \ref{sec: shadow} below) . However, the local mean variance criterion is, in general,
not monotone in wealth.\footnote{When $\gamma=1$, the local-mean variance objective agrees with logarithmic utility, for which monotonicity holds and the shadow market strategy is the optimal one cf.~\cite{gerhold2013dual}.} Therefore, verification of optimality fails, leaving open the question, whether trading strategies derived in the shadow market are also optimal in the original market with transaction costs.

\cite{gm2020} show that, in the presence of transaction costs, maximizing returns is well-posed even without controlling for volatility - transaction costs act as a penalty in the objective. As a consequence, the efficient frontier is not a straight line as in the classical Merton problem, but reaches a maximum for finite volatility, after which taking on even further risk may result in negative alpha. However, in frictionless markets, such an objective gives the incentive to seek arbitrary leverage, unless the asset has zero expected excess return. Thus shadow prices are destined to fail as an optimisation tool. 

Nevertheless, a candidate shadow price can be found for a risk-averse investor. A construction similar to \cite{gerhold2013dual} yields trading policies of the form \eqref{eq: trafosx}, thus they are indistinguishable at first order from the optimal one. Moreover, at second order, they are distinguished by a mere change of sign in the second order coefficient. Even more surprisingly, the equivalent safe rate of the shadow price trading strategy agree at third order with the maximum. In view of the second-order impact of transaction costs it is essentially optimal. However, we devise trading policies that strictly outperform the shadow price trading strategy at fourth order.

\subsection*{Program of Paper}

The paper is structured as follows: Section \ref{sec: the model} presents the market model, encompassing a risky Black-Scholes asset with transaction costs, the mean-variance objective, and a recap of the optimal trading policies established in \cite{gm2020}. Section \ref{sec: adm and perf} introduces control limit policies, evaluating their long-run performance along with small-transaction costs asymptotics (Lemma \ref{prop: performance controllable strategy}). In Section \ref{sec: shadow}, a na\"ive Ansatz for a shadow price is proposed, and asymptotic expansions of the trading boundaries are provided. Theorem \ref{lem perf optx} demonstrates their third order asymptotically optimality and Theorem \ref{thm: subopt} establishes their strict sub-optimality. Section \ref{sec: obstacle} provides a rigorous proof that for maximizing expected returns without controlling for volatility, no shadow price exists (Theorem \ref{th: no shadow price}). The final section \ref{sec: conclusion} summarizes our findings and points out directions for future research. The Appendix computes a high-order approximation of the candidate shadow price, to support the proof of Theorem \ref{thm: subopt}.

\section{Materials and Methods}\label{sec: the model}

The market model is comprised of two assets: a safe safe asset that is continuously compounded at a constant rate of $r\ge 0$ and a risky asset $S$ purchased at its ask price $S_t$ and satisfying the dynamics
\begin{equation}\label{eq: ask dynamics}
\frac{dS_t}{S_t}=(\mu+r)dt+\sigma dB_t, \quad S_0, \sigma,\mu > 0,
\end{equation}
where $B$ is a standard Brownian motion. The risky asset's bid (selling) price is $(1-\varepsilon)S_t$, which implies a constant relative bid-ask spread of $\varepsilon>0$, or, equivalently, constant proportional transaction costs.

Let $w$ be the wealth associated with a self-financing trading strategy\footnote{For the dynamics of the wealth process, see Lemma 2.3 below.}. The mean-variance trade-off is captured by maximizing 
the equivalent safe rate, 
\begin{align}\label{eq: obj asympt}
\esr &:= \limsup_{T\rightarrow\infty} \frac1 T \mathbb E\left[ \int_0^T \frac{dw_t}{w_t}-\frac\gamma 2 
\left\langle \int_0^\cdot 
\frac {dw_t}{w_t}\right\rangle_T
\right].
\end{align}
With $\pi$, the proportion of wealth invested in the risky asset, and with $\varphi_t$, the number of shares $\varphi_t=\varphi_t^\uparrow-\varphi_t^\downarrow$
being the difference of purchases $\varphi_t^\uparrow$ minus sales $\varphi_t^\downarrow$ one can rewrite the objective\footnote{This follows from the respective finite-horizon objective \eqref{eq main problem 1}, expressed in terms of $\pi_t$ and $\varphi_t$, see Lemma 2.3.}
as follows
\begin{align}
\esr :=&
r+\limsup_{T\rightarrow\infty}\frac{1}{T}\mathbb E\left[\int_0^T \left(\mu \pi_t-\frac{\gamma \sigma^2}{2}\pi_t^2\right)dt-\varepsilon\int_0^T\pi_t\frac{d\varphi^\downarrow_t}{\varphi_t} \right].
\end{align}
In the absence of transaction costs ($\varepsilon=0$), the objective is maximized by the constant-proportion portfolio $\pi_*:= \frac{\mu}{\gamma \sigma^2}$ dating back to Markowitz and Merton. The risk-neutral objective $\gamma = 0$ reduces to the average annualized return over a long horizon, which is well posed
for transaction costs \cite[Theorem 3.2]{gm2020}, but meaningless in the traditional framework with zero bid-ask spread, where a strategy can arbitrarily levered. The case $\gamma=1$ reduces to logarithmic utility, which is solved by the \cite{taksar1988diffusion} for the unlevered case $\frac{\mu}{\gamma\sigma^2}<1$.

An optimal strategy maximizing the equivalent safe rate exists. The following is a shortened version of \cite[Theorem 3.1]{gm2020}, characterising optimality:
\begin{theorem}\label{thm: gm20}
\noindent
Let $\frac{\mu}{\gamma\sigma^2}\neq 1$.
\begin{enumerate}
\item \label{thm free b part 1}
For any $\gamma>0$ there exists $\varepsilon_0>0$ such that for all $\varepsilon<\varepsilon_0$, 
there is a unique solution $(W,\zeta_-,\zeta_+)$, with $\zeta_-<\zeta_+$, for the free boundary problem 
\begin{align}\label{eq: TKA fbp}
&\textstyle{\frac{1}{2}\sigma^2 \zeta^2 W''(\zeta)+(\sigma^2+\mu)\zeta W'(\zeta)+\mu W(\zeta)-\frac{1}{(1+\zeta)^2}\left(\mu-\gamma\sigma^2\frac{\zeta}{1+\zeta}\right)=0,}\\\label{initial0 TKA}
&\textstyle{W(\zeta_-)=0}\\\label{initial1 TKA}
&\textstyle{W'(\zeta_-)=0,}\\\label{terminal0 TKA}
&\textstyle{W(\zeta_+)=\frac{\varepsilon}{(1+\zeta_+)(1+(1-\varepsilon)\zeta_+)},}\\
\label{terminal1 TKA}
&\textstyle{W'(\zeta_+)=\frac{\varepsilon(\varepsilon-2(1-\varepsilon)\zeta_+-2)}{(1+\zeta_+)^2(1+(1-\varepsilon)\zeta_+)^2}}
\end{align}

\item \label{thm free b part 2}
The trading strategy that buys at $\pi_- := \zeta_-/(1+\zeta_-)$ and sells at $\pi_+ := \zeta_+/(1+\zeta_+)$ as little as to keep the risky weight $\pi_t$ within the interval $[\pi_-,\pi_+]$ is optimal.

\item \label{thm free b part 3}
The maximum performance is 
\begin{equation}\label{eq: max perf}
\textstyle{\max_{\varphi\in\Phi}\lim_{T\rightarrow\infty}
\frac1T 
\mathbb E\left[
\int_0^T \left( \mu \pi_t - \frac{\gamma\sigma^2}{2} \pi_t^2 \right)dt
-{\varepsilon}\int_0^T\pi_t\frac{d\varphi^\downarrow_t}{\varphi_t} 
\right]
= \mu \pi_- -\frac{\gamma \sigma^2}{2}\pi_-^2,}
\end{equation}
where $\Phi$ is the set of admissible strategies in Definition \ref{def: admiss} below, $\varphi_t = \pi_t w_t/S_t$ is the number of shares held at time $t$, and $\varphi^\downarrow_t$ is the cumulative number of shares sold up to time $t$.

\item\label{thm free b part 4} 
The trading boundaries $\pi_-$ and $\pi_+$ have the asymptotic expansions
\begin{equation}\label{eq: trading boundaries TKA x}
\textstyle{\pi_\pm=\pi_*\pm\left(\frac{3}{4\gamma}\pi_*^2(\pi_*-1)^2\right)^{1/3}\varepsilon^{1/3}-\frac{(1-\gamma)\pi_*}{\gamma}\left(\frac{\gamma \pi_*(\pi_*-1)}{6}\right)^{1/3}\varepsilon^{2/3} +O(\varepsilon).}
\end{equation}
\item The equivalent safe rate ($\esr$) has the expansion
\begin{equation}\label{eq: asym esr}
\textstyle{\esr} =\textstyle{r+\frac{\gamma \sigma^2}{2}\pi_*^2-\frac{\gamma \sigma^2}{2}\left(\frac{3}{4 \gamma}\pi_*^2(\pi_*-1)^2\right)^{2/3}\varepsilon^{2/3}+ O(\varepsilon).}
\end{equation}
\end{enumerate}
\end{theorem}

\subsection{Admissible strategies and their long-run performance}\label{sec: adm and perf}
In view of transaction costs, only finite-variation trading strategies are consistent with solvency. This is illustrated
by the following example:
\begin{example}\label{ex: swap}
Consider the dynamic hedging part of $1/\varepsilon$ variance swaps\footnote{It is well-known that a variance swap with maturity $T$ on a continuous semimartingale $S$ can be perfectly hedged by holding $2/(T S_t)$ units of the underlying at time $t\leq T$ (the dynamic hedging term), and a static portfolio of European puts and calls with expiry $T$, \cite{vswaps}.} on the asset $S$ with maturity $T=2$, that requires to hold
\[
\varphi_t=\frac{1}{\varepsilon S_t}
\]
units of the underlying at each time $t\geq 0$. Trading discretely, along a mesh of size $\Delta$, one needs to sell at $t+\Delta$ if and only if $S_{t+\Delta}>S_t$, which incurs a cost of
\[
\varepsilon \times 1/\varepsilon \times S_{t+\Delta}(1/S_t-1/S_{t+\Delta})=(S_{t+\Delta}/S_t-1).
\]
Let $x_+=\max(0,x)$ and $N=T/\Delta$, then the total transaction cost amounts to
\[
C_N=\sum_{i=0}^{N-1} (S_{(i+1)\Delta}/S_{i\Delta}-1)_+.
\]
Note that this sum counts all positive simple returns of the asset, which can be approximated by logarithmic returns. Thus, as $N\rightarrow \infty$,
$C_N\rightarrow C$, the semivariation of a Brownian motion $B$ with drift, 
\[
C=\lim_{\Delta t\rightarrow 0}\sum_{i=0}^{T/\Delta t-1}(B_{(i+1)\Delta t}-B_{i \Delta t})_+=\infty,
\]
almost surely. This shows that, under proportional transaction costs, such a dynamic trading strategy results in immediate bankruptcy.
\end{example}
Denote by $X_t$ and $Y_t$ the wealth in the safe and risky positions respectively, and by $(\varphi_t^\uparrow)_{t\ge 0}$ and $(\varphi_t^\downarrow)_{t\ge 0}$ the cumulative number of shares bought and sold, respectively. The self-financing condition prescribes that $(X,Y)$ satisfy the dynamics
\begin{equation}\label{eq: sf1}
dX_t = rX_tdt-S_td\varphi^\uparrow_t+(1-\varepsilon) S_t d\varphi^\downarrow_t,\quad dY_t = S_t d\varphi^\uparrow_t- S_td \varphi^\downarrow_t+\varphi_t dS_t
.
\end{equation}

A strategy is admissible if it is non-anticipative and solvent, up to a small increase in the spread:
\begin{definition}\label{def: admiss}
Let $x>0$ (the initial capital) and let $(\varphi_t^\uparrow)_{t\ge 0}$ and $(\varphi_t^\downarrow)_{t\ge 0}$ be continuous, increasing processes, adapted to the augmented natural filtration of $B$. Then $(x, \;\varphi_t=\varphi_t^\uparrow-\varphi_t^\downarrow)$ is an \emph{admissible trading strategy} if 
\begin{enumerate}
\item its liquidation value is strictly positive at all times:
There exists $\varepsilon'>\varepsilon$ such that the discounted asset $\widetilde S_t:=e^{-rt}S_t$ satisfies
\begin{equation}\label{eq: liquidity}
x - \int_0^t \widetilde S_s d\varphi_s +\widetilde S_t \varphi_t - \varepsilon'\int_0^t \widetilde S_s d\varphi^\downarrow_s -\varepsilon'\varphi_t^+ \widetilde S_t > 0\qquad\text{a.s. for all }t\ge 0.
\end{equation}
 \item\label{part2 Lemma A21} 
The following integrability conditions hold:
\begin{equation}\label{eq: admissible trading}
\mathbb E\left[\int_0^t \vert\pi_u\vert^2 du\right]<\infty,\quad\mathbb E\left[\int_0^t \pi_u \frac{d\|\varphi_u\|}{\varphi_u}\right]<\infty \quad\text{ for all }t\ge 0,
\end{equation}
where $\|\varphi_t\|$ denotes the total variation of $\varphi$ on $[0,t]$. 
\end{enumerate}
The family of admissible trading strategies is denoted by $\Phi$.
\end{definition}

The following lemma describes the dynamics of the wealth process $w_t$, the risky weight $\pi_t$, and the risky-safe ratio $\zeta_t$. 
\begin{lemma}\label{le: rewriting obj fun}
For any admissible trading strategy $\varphi$:
\begin{align}\label{eq zeta diff}
\frac{d\zeta_t}{\zeta_t}&=\mu dt +\sigma dB_t+(1+\zeta_t)\frac{d\varphi_t^\uparrow}{\varphi_t}    -(1+(1-\varepsilon)\zeta_t)\frac{d\varphi_t^\downarrow}{\varphi_t},\\\label{eq w diff}
\frac{dw_t}{w_t}&=r dt+\pi_t(\mu dt+\sigma dB_t-\varepsilon \frac{d\varphi_t^\downarrow}{\varphi_t}),\\\label{eq pi diff}
\frac{d\pi_t}{\pi_t}&=(1-\pi_t)(\mu dt+\sigma dB_t)-\pi_t(1-\pi_t)\sigma^2 dt+\frac{d\varphi_t^\uparrow}{\varphi_t}-(1-\varepsilon\pi_t)\frac{d\varphi_t^\downarrow}{\varphi_t}.
\end{align}
For any such strategy, the functional
\begin{equation}\label{eq: finite mean var fun}
F_T(\varphi):=
\frac1 T \mathbb E\left[ \int_0^T \frac{dw_t}{w_t}-\frac\gamma 2\left\langle  \int_0^T \frac{d w_t}{w_t}\right\rangle_T\right]
\end{equation}
can be rewritten as
\begin{equation}\label{eq main problem 1}
F_T(\varphi)=r+\frac{1}{T}\mathbb E\left[\int_0^T \left(\mu \pi_t-\frac{\gamma \sigma^2}{2}\pi_t^2\right)dt-\varepsilon\int_0^T\pi_t\frac{d\varphi^\downarrow_t}{\varphi_t} \right].
\end{equation}
\end{lemma}
\begin{proof}
See \cite[Lemma A.2]{gm2020}.
\end{proof}

\begin{lemma}\label{lem: existence controllable strategy}
Let $\eta_-<\eta_+$ be such that either $\eta_+<-1/(1-\varepsilon)$ or $\eta_->0$. Then there exists an admissible trading strategy ${\hat\varphi}$ such that the risky-safe ratio $\eta_t$
satisfies SDE \eqref{eq zeta diff}. Moreover, $(\eta_t,{\hat\varphi}_t^\uparrow,{\hat\varphi}_t^\downarrow)$ is a reflected diffusion on
the interval $[\eta_-,\eta_+]$. In particular, $\eta_t$ has stationary density equals
\begin{equation}\label{eq st de 1}
 \nu(\eta):=\frac{\frac{2\mu}{\sigma^2}-1}{\eta_+^{\frac{2\mu}{\sigma^2}-1}-\eta_-^{\frac{2\mu}{\sigma^2}-1}}\eta^{\frac{2\mu}{\sigma^2}-2}, \quad \eta\in[\eta_-,\eta_+],
\end{equation}
when $\eta_->0$, and otherwise equals
\begin{equation}\label{eq st de 2}
 \nu(\eta):=\frac{\frac{2\mu}{\sigma^2}-1}{\vert\eta_-\vert^{\frac{2\mu}{\sigma^2}-1}-\vert\eta_+\vert^{\frac{2\mu}{\sigma^2}-1}}\vert\eta\vert^{\frac{2\mu}{\sigma^2}-2}, \quad \eta\in[\eta_-,\eta_+].
\end{equation}
\end{lemma}
\begin{proof}
See \cite[Lemma B.5]{gm2020}.
\end{proof}
\begin{definition}\label{def: controllable}
For the rest of the paper, the strategy in Lemma \ref{lem: existence controllable strategy} is called ``control limit policy for $\eta_\pm$", an adaption of the name of similar policies in \cite{taksar1988diffusion}, where ``limit" actually relates to the boundaries of the interval $[\eta_-,\eta_+]$. Note that the strategy in Theorem \ref{thm: gm20} \eqref{thm free b part 2} is exactly of this kind: It entails no trading, as long as $\zeta_t\in (\zeta_-,\zeta_+)$, and trades as little as necessary at $\zeta_\pm$ to keep the risky-safe ratio in the interval $[\zeta_-,\zeta_+]$. Alternative strategies, such as trading into the middle of the no-trade region,
incur significantly larger transaction costs.\footnote{By ergodicity, the strategy that makes bulk trades into the middle of the optimal no-trade region incurs average transaction costs of higher order, namely proportional to $\varepsilon^{1/3}$. (Compare the ATC \eqref{eq: optimal ATC} which is of second order.)  }
\end{definition}

The following computes the statistics contributing to the ESR of any trading strategy as in Lemma \ref{lem: existence controllable strategy} (not just the optimal one)
in terms of the risky-safe ratio.
\begin{lemma}\label{prop: performance controllable strategy}
Consider a control limit policy for $\eta_\pm$. Long-run mean $\hat m$, long-run standard deviation $\hat \sigma$ and average transaction costs ATC are given by the almost sure limits,
\begin{align}
\label{eq: m}
\hat{m} &=\lim_{T\rightarrow\infty }\frac{1}{T}\int_0^T\frac{dw_t}{w_t}dt= r+\mu \int_{\eta_-}^{\eta_+} \left(\frac{\zeta}{1+\zeta}\right)\nu(d\eta),\\
\label{eq: sigma}
\hat{\sigma}^2 &=\lim_{T\rightarrow\infty}\frac{1}{T}\left\langle\int_0^{\cdot}\frac{dw_t}{w_t}\right\rangle_T=  \sigma^2\int_{\eta_-}^{\eta_+} \left(\frac{\zeta}{1+\zeta}\right)^2\nu(d\eta),\\\label{eq: ATC}
\avtrco&=\varepsilon\lim_{T\rightarrow\infty}\frac{1}{T}\int_0^T \pi_t \frac{d\varphi_t^\downarrow}{\varphi_t}=\frac{\sigma^2(2\mu/\sigma^2-1)}{2} \left(\frac{\frac{\varepsilon \zeta_+}{(1+\zeta_+)(1+(1-\varepsilon)\zeta_+)}}{1-\left(\frac{{\zeta_-}}{{\zeta_+}}\right)^{2\mu/\sigma^2-1}}\right),
\end{align}
where $\nu$ is the stationary density of Lemma \ref{lem: existence controllable strategy}. 
\end{lemma}
\begin{proof}
All the formulae use the ergodic theorem and thus can be obtained with the methods in \cite{gm2020}. In particular, identity \eqref{eq: ATC} holds in view
of \cite[Lemma C.1]{gerhold.al.11}.
\end{proof}
Using the analytic expressions of \eqref{eq: m}, \eqref{eq: sigma} and \eqref{eq: ATC} with MATHEMATICA, we obtain explicit asymptotics, precise at third order in $\varepsilon^{1/3}$:
\begin{lemma}\label{lem perf opt}
For the optimal strategy of Theorem \ref{thm: gm20}, the statistics of Lemma \ref{prop: performance controllable strategy} satisfy the following asymptotics
\begin{align}\label{long run mean opt}
\hat m&=r+\frac{\mu^2}{\gamma \sigma^2}-\frac{\mu(2 \pi_*-1)}{\gamma}\left(\frac{\gamma \pi_*(\pi_*-1)}{6}\right)^{1/3}\varepsilon^{2/3}+O(\varepsilon^{4/3}),\\
\hat \sigma^2&=\frac{\mu^2}{\gamma^2\sigma^2}-\frac{\sigma^2 \pi_*(7 \pi_*-3)}{2\gamma}\left(\frac{\gamma \pi_*(\pi_*-1)}{6}\right)^{1/3}\varepsilon^{2/3}+O(\varepsilon^{4/3}),\\\label{eq: optimal ATC}
\avtrco&=\frac{3 \sigma^2}{\gamma}\left(\frac{\gamma\pi_*(\pi_*-1)}{6}\right)^{4/3}\varepsilon^{2/3}-\frac{\mu(\gamma-1)}{2\gamma} \pi_*(\pi_*-1) \varepsilon+ O(\varepsilon^{4/3}).
\end{align}
The maximum equivalent safe rate satisfies
\begin{equation}\label{eq: asym esrxx}
\textstyle{\esr} =\textstyle{r+\frac{\gamma \sigma^2}{2}\pi_*^2-\frac{\gamma \sigma^2}{2}\left(\frac{3}{4 \gamma}\pi_*^2(\pi_*-1)^2\right)^{2/3}\varepsilon^{2/3}+\frac{\mu(\gamma-1)}{2\gamma} \pi_*(\pi_*-1) \varepsilon+ O(\varepsilon^{4/3}).}
\end{equation}
\end{lemma}

\begin{remark}
\begin{enumerate}
\item All asymptotics of Lemma \ref{lem perf opt} improve those of \cite[Theorem 3.1]{gm2020} in precision by one order. Note that \eqref{long run mean opt} is a corrected version of \cite[Theorem 3.1, eq. (3.8)]{gm2020}, where the bracket $(5\pi_*-3)$ is given, instead of the correct term $(2\pi_*-1)$ in \eqref{long run mean opt}.
\item One can run a consistency check that compares the asymptotics \eqref{eq: asym esrxx} of the maximum ESR  (computed, by developing $r+\hat m-\frac{\gamma}{2}\hat \sigma^2-\text{ATC}$ into a formal power series in $\varepsilon^{1/3}$) with the asymptotic expansion of the shorter formula $r+\mu \pi_-\frac{\gamma \sigma^2}{2}\pi_-^2$ from Theorem \ref{thm: gm20} \ref{thm free b part 3}.
\end{enumerate}
\end{remark}
\section{Results}
\subsection{Asymptotically optimal shadow policies}\label{sec: shadow}

In this section, a shadow price for the mean-variance objective \eqref{eq: obj asympt} is constructed, and asymptotic formulas for the implied strategy, that is optimal in the shadow market, are derived. The exposition is motivated by the shadow price construction for log-utility investors, cf.~\cite[Chapter 3]{gerhold2013dual}, see also \cite{ gerhold.al.11,guasonimuhlekarbe}. Assume the following functional form of the shadow price $\widetilde S_t$,
\begin{equation}\label{eq: shadow ansatz}
\widetilde S_t= g(\pi_t) S_t,
\end{equation}
where $g$ satisfies the boundary conditions
\begin{equation}\label{eq: boundary 0}
g(\pi_-)=1,\quad g(\pi_+)=(1-\varepsilon),
\end{equation}
reflecting that an optimal strategy (such as of Theorem \ref{thm: gm20}) is a control limit policy for $\pi_\pm$, which buys (resp. sells) the frictionless asset $\widetilde S$ precisely when
its price equals the ask price $S$, and sells precisely, when it equals the bid price $(1-\varepsilon)S$. 

If $\widetilde S$ satisfies \eqref{eq: shadow ansatz} with twice differentiable $g$, and if $g$ satisfies
\eqref{eq: boundary 1}, then It\^o's formula yields the dynamics of instantaneous returns\footnote{The product rule gives \begin{equation}\label{sh 1}
\frac{d\widetilde S_t}{\widetilde S_t}=\frac{dS_t}{S_t}+\frac{dg}{g}+\frac{d\langle S_t, g\rangle }{S_t g}=:(\widetilde{\mu_t}+r)dt+\widetilde{\sigma}_t dB_t,
\end{equation}
from which the particular form of drift and diffusion coefficients \eqref{eq: drift coef}, \eqref{eq: diff coef} can be computed.}
\[
\frac{d\widetilde S_t}{\widetilde S_t}=rdt+d\widetilde \mu_t+\widetilde \sigma_t dB_t,
\]
with
\begin{align}\label{eq: drift coef}
d\widetilde{\mu}_t&=\mu dt+\frac{g'(\pi_t)\left(\pi_t(1-\pi_t)\mu dt +\pi_t(1-\pi_t)^2\sigma^2\right) }{g(\pi_t)}\\\nonumber
&+\frac{\frac{1}{2}g''(\pi_t)\pi_t^2(1-\pi_t)^2\sigma^2dt+g'(\pi_t)\left(\pi_t\frac{d\varphi_t^\uparrow}{\varphi_t} -(1-\varepsilon\pi_t)\frac{d\varphi_t^\downarrow}{\varphi_t}\right)}{g(\pi_t)}
\end{align}
and diffusion coefficient
\begin{equation}\label{eq: diff coef}
\widetilde{\sigma}_t=(\sigma g(\pi_t)+g'(\pi_t)\pi_t(1-\pi_t)\sigma)/g(\pi_t).
\end{equation}
The smooth pasting condition
\begin{equation}\label{eq: boundary 1}
g'(\pi_-)=g'(\pi_+)=0
\end{equation}
is imposed such that instantaneous drift of the shadow price become absolutely continuous (the condition removes the local time terms
$\frac{d\varphi_t^\uparrow}{\varphi_t}$ and $\frac{d\varphi_t^\downarrow}{\varphi_t}$), thus $d\widetilde \mu_t=\widetilde \mu_t dt$, with
\begin{equation}\label{sh 2}
\widetilde{\mu}_t=\mu+\frac{g'(\pi_t)(\pi_t(1-\pi_t)\mu +\pi_t(1-\pi_t)^2\sigma^2)+\frac{1}{2}g''(\pi_t)\pi_t^2(1-\pi_t)^2\sigma^2}{g(\pi_t)}.
\end{equation}

The fraction of wealth $\widetilde{\pi}$ invested in
the risky asset, evaluated at the shadow price, satisfies 
\begin{equation}\label{sh 3}
\widetilde \pi_t=\frac{Y_tg(\pi_t)}{X_t+Yg(\pi_t)}=\frac{\pi_t g(\pi_t)}{(1-\pi_t)+\pi_t g(\pi_t)}.
\end{equation}
Mean-variance optimality in the shadow market holds, when the proportion of wealth in the shadow market's risky asset $\widetilde S$ equals the Merton fraction, that is
\[
\widetilde \pi_t=\frac{\widetilde{\mu_t}}{\gamma \widetilde{\sigma}^2_t}.
\]
Equating this solution with \eqref{sh 3}, and using \eqref{sh 2}, \eqref{eq: diff coef} entails that $g$ satisfies the ODE
\begin{align}\label{eq: ODEmomo}
\frac{1}{2}g''(\pi)\pi^2(1-\pi)^2 \sigma^2&=\frac{\gamma \pi \sigma^2(g+g'(\pi)\pi(1-\pi))^2}{1-\pi+\pi g(\pi)}-\mu g(\pi)\\\nonumber
&-g'(\pi)(\pi(1-\pi)\mu +\pi(1-\pi)^2\sigma^2).
\end{align}
Define $\Psi$ implicitly as
\[
g(\pi)=\frac{\Psi(Y/X)}{Y/X}=:\frac{\Psi(\zeta)}{\zeta},
\]
and set
\begin{equation}\label{eq: trafo}
\zeta_{\pm}=\frac{\pi_\pm}{1-\pi_\pm}.
\end{equation}
Then ($\Psi$, $\zeta_-$, $\zeta_+$) satisfy the problem
\begin{align}\label{eq: fbp psi1}
\Psi'' (\zeta)&=\frac{2\gamma \Psi'^2(\zeta)}{(1+\Psi(\zeta))}-\frac{2\mu}{\sigma^2}\frac{  \Psi'(\zeta)}{\zeta},\\\label{eq: fbp psi2}
\Psi'(\zeta_-)&= \Psi(\zeta_-)/\zeta_-=1,\\\label{eq: fbp psi3}
\Psi'(\zeta_+)&=\Psi(\zeta_+)/\zeta_+=(1-\varepsilon).
\end{align}
This is a free boundary problem, because both $\Psi$ and the trading boundaries $\zeta_\pm$ for the control limit policy are unknown.

Using the explicit solution $\Psi$ of the corresponding initial value problem \eqref{eq: fbp psi1}--\eqref{eq: fbp psi2}, and respecting terminal conditions \eqref{eq: fbp psi3}, one obtains a non-linear system of equations for $({\zeta}_-,{\zeta}_+)$. For small $\varepsilon$, this very system allows a unique solution
with asymptotic expansion\footnote{For the details leading to this and other asymptotics, see Appendix \ref{app: A}, Proposition \ref{th free} and Remark \ref{rem: free}.}
\begin{equation}\label{eq: asymptotics zeta shadow}
\widetilde{\zeta}_\pm=\frac{\pi_*}{1-\pi_*}\pm\left(\frac{3}{4\gamma}\right)^{1/3}\left(\frac{(\pi_*)^2}{1-\pi_*}\right)^{2/3} \varepsilon^{1/3}-\frac{(1+2\gamma)\pi_*}{2\gamma(1-\pi_*)^2}\left(\frac{\gamma \pi_*(\pi_*-1)}{6}\right)^{1/3}\varepsilon^{2/3}+O(\varepsilon).
\end{equation}
In comparison, the optimal strategy of Theorem \ref{thm: gm20} is a control limit policy whose limits $\zeta_\pm$, in terms of the risk-safe ratio have the expansion\footnote{This expression is readily obtained from \eqref{eq: trading boundaries TKA x} by expanding \eqref{eq: trafo} into formal power series in $\varepsilon^{1/3}$.}
\begin{equation}\label{eq: asymptotics zeta}
\textstyle{\zeta_\pm=\frac{\pi_*}{1-\pi_*}\pm\left(\frac{3}{4\gamma}\right)^{1/3}\left(\frac{\pi_*}{(\pi_*-1)^2}\right)^{2/3} \varepsilon^{1/3}-\frac{(5-2\gamma)\pi_*}{2\gamma(\pi_*-1)^2}\left(\frac{\gamma \pi_*(\pi_*-1)}{6}\right)^{1/3}\varepsilon^{2/3}+O(\varepsilon).}
\end{equation}

Note the factor $(1+2\gamma)$ in the $\varepsilon^{2/3}$ term in \eqref{eq: asymptotics zeta shadow}, which differs from the factor $(5-2\gamma)$ in \eqref{eq: asymptotics zeta}. Accordingly, the associated trading boundaries have an asymptotic expansion,
\begin{equation}\label{eq: trading boundaries shadow}
\widetilde{\pi}_\pm=\pi_*\pm\left(\frac{3}{4\gamma}(\pi_*)^2(1-\pi_*)^2\right)^{1/3}\varepsilon^{1/3}
+\frac{(1-\gamma)\pi_*}{\gamma}\left(\frac{\gamma \pi_*(\pi_*-1)}{6}\right)^{1/3}\varepsilon^{2/3}+O(\varepsilon).
\end{equation}

The expansion \eqref{eq: asymptotics zeta shadow} resp. \eqref{eq: trading boundaries shadow} agrees with the above expansions \eqref{eq: asymptotics zeta} resp. \eqref{eq: trading boundaries TKA x} up to first order (they agree in constant and in $\varepsilon^{1/3}$ terms). But they disagree in a quite subtle way for any $\gamma\neq 1$ at second order: The absolute values, but not the signs of second order term of $\pi_\pm$  (see \eqref{eq: trading boundaries shadow}) and $\widetilde \pi_\pm$ (see \eqref{eq: trading boundaries TKA x}) are identical.

The following establishes asymptotic optimality of third order of the strategy obtained from the shadow market (The proof exclusively uses MATHEMATICA and higher order expansions of \eqref{eq: trading boundaries shadow}).
\begin{theorem}\label{lem perf optx}
The asymptotic expansion of the equivalent safe rate and average transaction costs of the control limit policy for $\widetilde \pi_\pm$ are of the exact same form
as \eqref{eq: asym esrxx} resp. \eqref{eq: ATC}. Thus the strategy is asymptotically optimal at third order.  Long run mean and variance defined by Lemma \ref{prop: performance controllable strategy} satisfy the following asymptotics:
\begin{align}\label{long run mean optx}
\widetilde m&=r+\frac{\mu^2}{\gamma \sigma^2}-\frac{\mu(2 \gamma\pi_*-1)}{\gamma}\left(\frac{\gamma \pi_*(\pi_*-1)}{6}\right)^{1/3}\varepsilon^{2/3}+O(\varepsilon^{4/3}),\\\label{long run vol optx}
\widetilde \sigma^2&=\frac{\mu^2}{\gamma^2\sigma^2}-\frac{\sigma^2 \pi_*(\pi_*(8 \gamma+1)-3)}{2\gamma}\left(\frac{\gamma \pi_*(\pi_*-1)}{6}\right)^{1/3}\varepsilon^{2/3}+O(\varepsilon^{4/3}).
\end{align}
\end{theorem}
\begin{remark}
Note that, almost miraculously,
\[
\hat m-\frac{\gamma}{2}\hat \sigma^2=\widetilde m-\frac{\gamma}{2}\widetilde{\sigma}^2=O(\varepsilon^{4/3})
\]
because average transaction costs as well as the equivalent safe rate agree for both strategies up to third order, and mean and variance's third order terms vanish (compare \eqref{long run mean optx}--\eqref{long run vol optx} with mean and variance of the optimal strategy in Lemma \ref{lem perf opt}).
\end{remark}

\subsection{Outperforming the shadow market}\label{sec: shadow not optimal}
In Theorem \ref{lem perf optx} it has been shown that $\widetilde S$ is an asymptotic shadow price, as the strategy that is optimal in the frictionless market, is even optimal at third order in the original market. For the proof of this statement, it was crucial to have precise asymptotic expansions of the trading boundaries $\widetilde\pi_\pm$. 

The objective of this section is to prove that this strategy is not optimal. To this end, it would be useful to have higher (fourth and fifth) order terms in the expansion of the the optimal trading boundaries $\zeta_\pm$ and thus the maximum performance \eqref{eq: asym esrxx}. However, the free boundary problem of \eqref{new fbvp1x}--\eqref{new fbvp4x} associated with the optimal solution of Theorem \ref{thm: gm20} is notoriously difficult to deal with, even with MATHEMATICA, while the free boundary problem \eqref{eq: its a boy}--\eqref{eq: baby is born} arising from the shadow price Ansatz is much more tractable. Therefore, instead of developing the maximum performance to even higher precision, a strategy is found that merely outperforms the shadow market:
\begin{theorem}\label{thm: subopt}
Suppose $\gamma\notin\{0,1\}$. For any $\theta\in\mathbb R$, the family of control limit policies for
\begin{equation}\label{eq: kappa factor}
\widetilde\pi_\pm^\theta:=\widetilde\pi_\pm+(\theta-1)\times \frac{(\gamma-1)(\pi_*)^2(1-\pi_*)}{6}  \left(\frac{6}{\gamma\pi_*(1-\pi_*)}\right)^{2/3}\varepsilon^{2/3}
\end{equation}
has equivalent safe rate
\begin{align}\label{eq: asym esrxxx}
\textstyle{\esr}& =r+\frac{\gamma \sigma^2}{2}\pi_*^2-\frac{\gamma \sigma^2}{2}\left(\frac{3}{4 \gamma}\pi_*^2(\pi_*-1)^2\right)^{2/3}\varepsilon^{2/3}+\frac{\mu(\gamma-1)}{2\gamma} \pi_*(\pi_*-1) \varepsilon\\\nonumber&-\frac{\sigma^2\times  k(\theta)}{20\gamma}\left(\frac{\gamma\pi_*(1-\pi_*)}{6}\right)^{2/3}\varepsilon^{4/3}+O(\varepsilon^{5/3}),
\end{align}
where
\[
k(\theta):=-9 + 2  \pi_* \left(9 + \pi_*\left(3 + 12\gamma ( \gamma-2)  + (10\theta+5 \theta^2) (\gamma-1)^2\right)\right)
\]
and thus is asymptotically optimal at third order. For sufficiently small $\varepsilon$, the best performance at fourth order is achieved for $\theta=-1$, strictly outperforming
the shadow performance ($\theta=1$).
\end{theorem}
\begin{proof}
Using the method in Appendix \ref{app: A}, derive asymptotic expansions of $c$ and $s$ (whence of $\widetilde\zeta_\pm$ up to sixth order), satisfying
the free boundary problem \eqref{new fbvp1x}--\eqref{new fbvp4x} at the same order. Modifying the second order term by including a factor $\theta$ as in \eqref{eq: kappa factor},
one arrives at \eqref{eq: asym esrxxx}. The fourth order coefficient $k(\theta)$ is a polynomial of second order in $\theta$, with global minimum at $\theta=-1$. Comparison sign and magnitude of
this factor is straightforward and reveals that $\theta=-1$ outperforms any other control limit policy for $\theta\neq -1$.
\end{proof}
\begin{remark}
Note that for $\theta=-1$, the control limit policy for \eqref{eq: kappa factor} is, up to order 2, equal to the optimal strategy \eqref{eq: trading boundaries TKA x} of Theorem \ref{thm: gm20}. This does not mean that it is optimal at order four or beyond, as higher order coefficients of $\widetilde \pi_\pm^\theta$ may not agree with the those of
the optimal boundaries $\pi_\pm$.
\end{remark}
\subsection{The limits of shadow prices}\label{sec: obstacle}
Recall that a shadow price $\widetilde S$ is a frictionless process evolving in the bid-ask spread
\begin{equation}\label{eq: shadow bidask}.
(1-\varepsilon)S_t\leq \widetilde S_t\leq S_t, \quad t\geq 0
\end{equation}
such that the optimal strategy $\varphi$ is also optimal in the original market, and buys (resp. sells) precisely when $\widetilde S_t=S_t$ (resp. $\widetilde S_t=(1-\varepsilon)S_t$).

To start with, the dynamics of risky-safe ratio, wealth and proportion of wealth in the shadow market, for any finite variation strategy, is stated.
\begin{lemma}\label{le: shadow strategies}
Suppose the shadow price satisfies the dynamics
\begin{equation}\label{eq: shadow SDE}
\frac{d\widetilde S_t}{\widetilde S_t}=(r+\widetilde \mu_t)dt+\widetilde \sigma_t dB_t.
\end{equation}
For any finite variation trading strategy $\varphi$,
\begin{align}\label{eq zeta diffx}
\frac{d\widetilde\zeta_t}{\widetilde\zeta_t}&=\widetilde\mu_t dt +\widetilde\sigma_t dB_t+(1+\widetilde\zeta_t)\frac{d\varphi_t^\uparrow}{\varphi_t}    -(1+\widetilde\zeta_t)\frac{d\varphi_t^\downarrow}{\varphi_t},\\\label{eq w diffx}
\frac{d\widetilde w_t}{\widetilde w_t}&=r dt+\widetilde\pi_t(\widetilde\mu_t dt+\widetilde\sigma_t dB_t),\\\label{eq pi diffx}
\frac{d\widetilde \pi_t}{\widetilde\pi_t}&=(1-\widetilde\pi_t)(\widetilde\mu_t dt+\widetilde\sigma_t dB_t)-\widetilde\pi_t(1-\widetilde\pi_t)\widetilde\sigma_t^2 dt+\frac{d\varphi_t^\uparrow}{\varphi_t}-\frac{d\varphi_t^\downarrow}{\varphi_t}.
\end{align}
\end{lemma}
\begin{proof}
A similar proof as that of Lemma \ref{le: rewriting obj fun} applies.
\end{proof}

\begin{lemma}\label{lem: shadow}
If a shadow price exists, then for the optimal strategy, cash position in original and shadow market agree ( $\widetilde X=X$) and the fraction of wealth invested in the shadow price satisfies
\[
\pi_t\leq\widetilde \pi_t\leq \frac{(1-\varepsilon)\pi_t }{1-\varepsilon\pi_t}.
\]
In particular, if the optimal strategy satisfies $\pi_t\in [\pi_-,\pi_+]$, then
\begin{equation}\label{eq: superbounds}
\pi_-\leq \widetilde \pi_t\leq \frac{(1-\varepsilon)\pi_+}{1-\varepsilon\pi_+}
\leq (1-\varepsilon)\pi_+(1+\varepsilon\pi_+).
\end{equation}
\end{lemma}
\begin{proof}
The optimal strategy trades the risky asset at the same prices in both markets, therefore the cash positions agree.

The lower bound is proved by observing that for $a,b>0$, the function 
\[
\frac{a\xi}{-b+a\xi}
\]
is strictly decreasing for any $\xi>b/a$ (which corresponds to positive wealth), and since $\widetilde S\leq S$, 
\[
\widetilde \pi_t=\frac{\varphi_t \widetilde S_t}{X_t+\varphi_t \widetilde S_t}\geq \frac{\varphi_t  S_t}{X_t+\varphi_t  S_t}=\frac{ \varphi_t S_t}{w_t}= \pi_t.
\]
Similarly, the upper bound follows from
\[
\widetilde \pi_t=\frac{\varphi_t \widetilde S_t}{X_t+\varphi_t \widetilde S_t}\leq \frac{(1-\varepsilon)\varphi_t S_t}{X_t+\varphi_t  (1-\varepsilon) S_t}=\frac{(1-\varepsilon)\pi_t }{1-\varepsilon\pi_t}.
\]
The constant bounds in terms of the trading boundaries $\pi_\pm$ are an obvious conclusion. The last inequality in \eqref{eq: superbounds} follows from the summation formula of the geometric series, knowing that solvency implies $\varepsilon \pi_+<1$.
\end{proof}

For proportional transaction costs, maximizing expected excess returns
\[
\lim_{T\rightarrow \infty}\frac{1}{T}\mathbb E\left[\int_0^T\frac{dw_t}{w_t}\right]
\]
over all admissible strategies $\varphi\in\Phi$, is well posed. By
\cite[Theorem 3.2]{gm2023},
for sufficiently small $\varepsilon$ there exists $0<\pi_-<\pi_+<\infty$ such that the trading strategy $\hat\varphi$ that buys at $\pi_- $ and sells at $\pi_+ $ to keep the risky weight $\pi_t$ within the interval $[\pi_-,\pi_+]$ is optimal. The maximum expected return of this optimal strategy is given by the almost sure limit
\begin{equation}
\lim_{T\rightarrow\infty}
\frac1 T \int_0^T \frac{dw_t}{w_t}
= r + \mu \pi_- 
\end{equation}
and the trading boundaries have the series expansions
\begin{align}\label{as multiplier piminus}
\pi_- =& (1-\kappa)\kappa^{1/2}\left(\frac{\mu}{\sigma^2}\right)^{1/2}\varepsilon^{-1/2} + 1 + O(\varepsilon^{1/2}),\\
\pi_+ =& \kappa^{1/2}\left(\frac{\mu}{\sigma^2}\right)^{1/2}\varepsilon^{-1/2} + 1 + O(\varepsilon^{1/2}),
\end{align}
where $\kappa\approx 0.5828$ is the unique solution of
\[
\frac{3}{2}\xi+\log(1-\xi)=0,\quad \xi\in(0,1).
\]
The remainder of this section is dedicated to showing that a shadow market does not exist.

For technical reasons, it is assumed in this section that any shadow price satisfies the following
\begin{assumption}\label{eq: Assx}
A shadow price $\widetilde S$ is a continuous process satisfying the dynamics \eqref{eq: shadow SDE} with drift and diffusion coefficients being ergodic in the sense that, almost surely, for some $\bar \mu,\bar\sigma^2\in\mathbb R$,
\begin{equation}\label{eq: ergodic shadow}
\lim_{T\rightarrow \infty}\frac{1}{T} \int_0^T\widetilde \mu_t dt=\bar \mu,\quad \lim_{T\rightarrow \infty}\frac{1}{T}\int_0^T \widetilde \sigma_t^2 dt=\bar \sigma^2.
\end{equation}
\end{assumption}
\begin{remark}
The fairly general Assumption \ref{eq: Assx} is natural in that it applies to all known constructions of shadow prices in continuous-time models. In fact, typically the ratio $\frac{\widetilde S_t}{S_t}$ is equal to $g(\pi_t)$, where $g$ is a real analytic function and $\pi_t$ is a stationary process: the optimal proportion of wealth in the risky asset, evolving within an interval $[\pi_-,\pi_+]$, where one buys (resp. sells) precisely
at the trading boundary $\pi_i$ (resp. $\pi_+$) and satisfying $g(\pi_-)=1$ and $g(\pi_+)=1-\varepsilon$, reflecting the very definition of shadow price, agreeing
with ask (resp. bid) whenever shares are purchased (resp. sold). As these functions in the literature are all analytic, one can use It\^o's formula to derive the dynamics \eqref{eq: shadow SDE} in more explicit form:\footnote{The general form of drift and diffusion coefficients follows from the typical smooth pasting conditions $g'(\pi_\pm)=0$, along the same arguments as in Section \ref{sec: shadow} that turn \eqref{eq: drift coef} into \eqref{sh 2}, by removing local-time terms.} There exist continuous functions $h$, and $H$ such that
\[
\widetilde \mu_t=h(\pi_t),\quad \widetilde \sigma_t^2=H(\pi_t)
\]
By the ergodic theorem \cite[II.35 and II.36]{borodin2002handbook}, one obtains
the finite limits in \eqref{eq: ergodic shadow}.

\end{remark}
\begin{lemma}\label{lem: impossible mu}
If $\mu>\sigma^2/2$, then $\bar \mu>\bar\sigma^2/2$. In particular, $\bar \mu\neq 0$. 
\end{lemma}
\begin{proof}
Write the fraction $\widetilde S_t/S_t$ as explicit solutions, where the accrual factor $e^{rt}$ factors out. As $\mu>\sigma^2/2$, by the law of iterated logarithm, $e^{-rt}S_t$ almost surely tends to $\infty$ as $t\rightarrow \infty$. If $\bar \mu\leq \bar\sigma^2/2$, then for sufficiently large $t$, $\widetilde S_t$ behaves like a geometric Brownian motion with drift $\bar\mu$ and volatility $\bar \sigma$, whence by a similar argument, $\lim\inf_{t\rightarrow \infty}\widetilde S_t=0$, almost surely. Thus $\lim\inf_{t\rightarrow \infty}\widetilde S_t/S_t=0<(1-\varepsilon)$, a contradiction to \eqref{eq: shadow bidask}.
\end{proof}

\begin{theorem}\label{th: no shadow price}
If  $\mu>\sigma^2/2$, then a shadow price satisfying the dynamics \eqref{eq: shadow SDE}  Assumption \ref{eq: Assx} does not exist.
\end{theorem}
\begin{proof}
Assume, for a contradiction, there exists a shadow price $\widetilde S$. By Lemma \ref{le: shadow strategies} the shadow wealth $\widetilde w_t=\varphi_t \widetilde S_t+\widetilde X_t$ satisfies the SDE \eqref{eq w diffx}. Furthermore, by Lemma \ref{lem: shadow}
\[
\int_0^t \widetilde \pi_t^2\widetilde \sigma_u^2 du\leq (1-\varepsilon)^2 (\pi_+^2(1+\varepsilon\pi_+)^2 \int _0^t \widetilde \sigma^2_u du<\infty
\]
and thus the integral of the Brownian term is a martingale by Assumption \ref{eq: Assx}. Thus the strategy $\varphi$ with associated wealth $\widetilde w$ achieves its optimum
at
\[
\lambda:=\lim_{T\rightarrow \infty}\frac{1}{T}\mathbb E\left[\int_0^T\frac{d\widetilde w_t}{\widetilde w_t}\right]=r+\lim_{T\rightarrow \infty}\frac{1}{T}\mathbb E\left[ \int_0^T \widetilde \mu_t \widetilde \pi_t dt\right].
\]
Note that, by Assumption \ref{eq: Assx}, 
\[
\bar \mu=\lim_{T\rightarrow \infty}\frac{1}{T}\int_0^T \widetilde \mu_t dt=\lim_{T\rightarrow \infty}\frac{1}{T}\mathbb E\left[\int_0^T \widetilde \mu_t dt\right]
\]
and by Lemma \ref{lem: impossible mu}, $\bar\mu>0$. Furthermore, by Lemma \ref{lem: shadow}, there exist $0<L<U<\infty$ such that
\begin{equation}\label{eq: bind bound}
L \bar\mu<\lim_{T\rightarrow \infty}\frac{1}{T}\mathbb E\left[ \int_0^T \widetilde \mu_t \widetilde \pi_t dt\right]<U\bar\mu.
\end{equation}
Any alternative strategy $\varphi^\star$, whose proportion of wealth in the shadow price satisfies 
\begin{equation}\label{eq: out cond}
\widetilde \pi^*\geq U
\end{equation}
outperforms $\varphi$, because 
\[
\lambda^*\geq r+U\bar \mu>\lambda.
\]
Trading strategies that keep the exposure in the shadow asset constant to $U$ exist, but they are of infinite variation. To obtain a finite variation strategies satisfying \eqref{eq: out cond}, recall that by Lemma \ref{le: shadow strategies}, the fraction of wealth in the shadow asset $\widetilde w$ associated with a finite variation strategy  satisfies \eqref{eq pi diffx}. One can modify this strategy, by allowing bulk trades: Let $\varphi^*$ be the the finite variations strategy that does refrain from trading, whenever $\widetilde \pi^\star \in (U,2U)$ but buys (resp. sells) the shadow asset in bulk, whenever $\widetilde \pi^\star$ hits  $U$ (resp. $2U$) so to reset $\widetilde \pi^\star$ to the midpoint $3U/2$. Such a strategy can be constructed pathwise, and satisfies
\[
\frac{d\widetilde \pi_{t-}^\star}{\widetilde\pi_{t-}^\star}=(1-\widetilde\pi_{t-}^\star)(\widetilde\mu_t dt+\widetilde\sigma_t dB_t)-\widetilde\pi_{t-}^\star(1-\widetilde\pi_{t-}^\star)\widetilde\sigma_t^2 dt+\frac{\Delta\varphi_t^{\star,\uparrow}}{\varphi_t^\star}-\frac{\Delta\varphi_t^{\star,\downarrow}}{\varphi_t^\star}.
\]
The existence of such a strategy contradicts optimality of $\varphi$, and thus a shadow price does not exist.
\end{proof}
\begin{remark}
The finite variation strategy in the end of the proof cannot be replaced by a (standard) reflected diffusion with two reflecting boundaries, because for the existence of strong solutions to the associated SDE on convex domains \cite[Theorem 4.1]{tanaka}, one would need $\widetilde \mu_t$ and $\widetilde \sigma_t$ be regular enough functions of $\widetilde \pi_t^\star$, an assumption to strong in this context. Also, it is unknown, whether such as strategy is solvent in the original market with transaction costs.
\end{remark}

\section{Discussion}\label{sec: conclusion}
Optimizing portfolios in continuous-time markets with proportional costs presents mathematically challenging problems. Strategies that are optimal in frictionless markets must be adjusted to prevent immediate bankruptcy, as exemplified by the dynamic hedging component of a variance swap (see Example \ref{ex: swap}). The strategies considered in this paper are stationary\footnote{More precisely, certain portfolio statistics, such as $\pi_t$ or $\zeta_t$, exhibit stationarity.} and thus ergodic theorems are used to determine their long-run performance. To gain insights into trading frequency, transaction costs, and long-run performance, we derive asymptotic expansions of the trading boundaries for small bid-ask spread.

The paper explores the (candidate) shadow prices for local mean-variance investors, with a threefold contribution:

First, we discover that the optimal strategy in the (candidate) shadow market differs from the optimal one in the original market, but only in the second-order terms of the asymptotic expansion of the trading boundaries.\footnote{That this second order discrepancy is not essential, can be seen also by a numerical robustness check, with trades at daily frequency and with a finite time horizon of, say five years. Numerical examples are already elaborated for a similar objective in great detail in \cite[Section 6 (Figures 4 and 5)]{gm2023}.}. Theorem \ref{lem perf optx} demonstrates that, for risk aversion $\gamma > 0$, the equivalent safe rate of the shadow market strategy agrees at the third order with the maximum. As transaction costs are of second order, we conclude that the performance of the shadow market strategy is essentially optimal. It is worth noting that the same is true\footnote{This assertion can be proven using the same method as in Theorem \ref{lem perf optx}.} for a long-run power-utility investor (cf.~\cite{gerhold.al.11}), as their trading boundaries also agree at the first order with \eqref{eq: trafosx}. Second, Theorem \ref{thm: subopt} establishes that for $\gamma\neq 1$, the (potential) shadow market strategy $\widetilde \pi$ is not optimal, as it can be outperformed. The alternative strategy is not necessarily optimal, even though it agrees up to the second order with the optimal one. In summary, the (candidate) shadow price is an asymptotic shadow price. Third, Theorem \ref{th: no shadow price} demonstrates that for risk-neutral investors ($\gamma=0$), no such shadow market exists.

The findings of this paper prompt the following research problems. First, we conjecture that a minor modification of the objective will render the shadow price candidate of Section \ref{sec: shadow} optimal in the original market with transaction costs. Motivated by
\cite{martin1, martin2}, we propose to replace the equivalent safe rate in \eqref{eq: obj asympt} by an infinite horizon, local-mean variance utility function\footnote{Note that we use portfolio returns, as opposed to changes of wealth in \cite{martin1, martin2}. Besides, Martin's work cares about asymptotic optimality at lowest order, similar to \cite{kallsenmad}.}
\begin{align}\label{eq: obj asymptx}
\esr &:= \mathbb E\left[\int_0^\infty \delta e^{-\delta t}\frac{dw_t}{w_t}-\frac\gamma 2 
 \int_0^\infty \delta e^{-\delta t}\left\langle
\frac {dw_t}{w_t}\right\rangle_t
\right]
\end{align}
for some discount rate $\delta>0$. In the absence of transaction costs ($\varepsilon=0$), the maximum equivalent safe rate agrees with that of the old objective \eqref{eq: obj asympt}. More importantly, this objective leads to the exact same shadow market construction, as in Section \ref{sec: shadow}. The question remains if our shadow market policy maximizes also \eqref{eq: obj asymptx} in the original market, surpassing its third order optimality (Theorem \ref{lem perf optx}). 

Second, the mathematical treatment of optimization problems involving transaction costs is always uniquely tailored to a specific objective. This results in free boundary problems that vary significantly, encompassing scenarios from Riccati Differential Equations \cite{gerhold2012asymptotics} and linear equations \cite[Theorem 3.3]{gm2020} to the nonlinear problem \eqref{eq: ODEmomo} addressed in this paper, and even singular problems for zero risk-aversion \cite[Theorem 3.2]{gm2020}. The question persists: Can a unified approach be devised that accommodates a diverse range of objectives? To explore this possibility, one might aim for conformity to a common format—a second-order free boundary problem stated as follows:
\begin{align}\label{eq: ODE2}
& F(g,g',g'')=0,\\\label{eq: shadow constraints}
g(\pi_-)&=1,\quad g(\pi_+)=1-\varepsilon,\\\label{eq: smooth pasting}
g'(\pi_-)&=0, \quad g'(\pi_+)=0.
\end{align}
This problem involves an unknown scalar function $g=g(\pi)$ that must satisfy a second-order nonlinear ODE \eqref{eq: ODE2}, along with buy and sell boundaries $\pi_-$ and $\pi_+$, respectively. The latter boundaries must adhere to zeroth-order boundary conditions \eqref{eq: shadow constraints} and first-order conditions \eqref{eq: smooth pasting}.\footnote{Such a general representation bears the advantage that
the stochastic process $\widetilde S_t:=g(\pi_t)S_t$ could be interpreted as a (candidate) shadow price.} In practical trading applications, a second-order approximation of the trading boundaries would suffice. Such approximation might be achieved through a general polynomial Ansatz for an approximation of \eqref{eq: ODE2}. 

Third, most of the literature\footnote{\cite[Theorem 6.16]{taksar1988diffusion} appears to be an exception, which does not refer to smallness of transaction costs.} regarding the existence of an optimal strategy and its asymptotic expansion, depends on the assumption of a ``sufficiently small" bid-ask spread $\varepsilon$, without providing a minimum $\varepsilon_0$, for which these statements hold. Are they applicable to actual bid-ask spreads observed in markets (for liquid assets ranging  in the basis points)? Addressing this question involves either demonstrating optimality for all $\varepsilon \in (0,1)$ or identifying counterexamples where optimality breaks down for larger transaction costs, along with determining the explicit lower bound $\varepsilon_0$ at which control limit policies remain optimal. Such lower bound would be contingent on model parameters $\gamma$, $\mu/\sigma^2$, and risk aversion. Most likely, it will depend on the chosen objective.

\appendix

\section{The free boundary problem for the shadow price candidate}\label{app: A}

Let us introduce two new parameters $c=1/\zeta_-$ and $s=\zeta_+/\zeta_-$. By defining the new function $\phi$ implicitly via
\[
\Psi(\zeta ):=\phi (c \zeta)/c,
\]
the free boundary problem \eqref{eq: fbp psi1}--\eqref{eq: fbp psi3} produces a similar one for $\phi(z)$, 
\begin{align}\label{new fbvp1x}
\phi'' (z)&=\frac{2\gamma \phi'^2(z)}{(c+\phi(z))}-2\gamma\pi^*\frac{  \phi'(z)}{z},\\
\label{new fbvp3x} \phi(1)&=1,\\
\label{new fbvp5x} \phi'(1)&=1,\\
\label{new fbvp2x}\phi(s)&=(1-\varepsilon)s,\\
\label{new fbvp4x} \phi'(s)&=(1-\varepsilon).
\end{align}
\begin{remark}\rm
Since for small transaction costs, trading strategies will be control limit policies on sufficiently small intervals, only the following cases need to be distinguished:
\begin{itemize}
\item $\zeta_-<\zeta_+<-1$ (levered case): Then $c<0$ and therefore the \ $z>0$, so $s<1$. Conversely, $s<1$ implies $\zeta_-<-1$.
\item $0<\zeta_-<\zeta_+$ (unlevered case): Then $c>0$ and therefore the argument $z<0$, so $s>1$. Conversely, $s>1$ implies $\zeta_->0$.
\end{itemize}
\end{remark}
For the sake of brevity, let us only consider the levered case, that is, $\zeta_-<\zeta_+<-1$ and $\phi(\zeta)<-1$. Since also $c<0$, one gets $\phi(z)>- c$ for all $z$. Also, $c+1>0$. Dividing \eqref{new fbvp1x} by $\phi'$ and integrating once, one thus obtains 
\[
\log(\phi' (z))=2\gamma\log(c+\phi(z))-2\gamma\pi^*\log z-2\gamma\log(c+1),
\]
where the initial condition  \eqref{new fbvp5x} has been respected. Taking antilogarithms, one thus obtains
\begin{equation}\label{eq: grandpa}
\frac{\phi'(z)}{(c+\phi(z))^{2\gamma}}=\frac{z^{-2\gamma\pi^*}}{(c+1)^{2\gamma}}.
\end{equation}
Exclude in the following the singular cases $\gamma\neq 1/2$ and $\mu/\sigma^2\neq 1/2$ (those cases can be dealt individually, leading to
simpler solutions of the ODE \eqref{eq: grandpa}.). Integrating once again, one gets
\begin{equation}\label{eq: precessor}
\frac{(c+\phi(z))^{1-2\gamma}}{1-2\gamma}=\frac{(z^{1-2\gamma\pi^*}-1)}{(1-2\gamma\pi^*)(c+1)^{2\gamma}}+\frac{(c+1)^{1-2\gamma}}{1-2\gamma},
\end{equation}
where the initial condition \eqref{new fbvp3x} has been respected. Thus
\begin{equation}\label{eq: sol in c}
\phi(z)=-c+\left(\frac{\frac{(1-2\gamma)}{(1-2\gamma\pi^*)}(z^{1-2\gamma\pi^*}-1)}{(c+1)^{2\gamma}}+(c+1)^{1-2\gamma}\right)^{\frac{1}{1-2\gamma}}.
\end{equation}
Until this stage, the terminal boundary conditions  \eqref{new fbvp2x} and \eqref{new fbvp4x} have not been involved. Those allow to reformulate the free boundary problem in terms of a system of non-linear equations for $s$ and $c$:

\begin{lemma}\label{nec s c}
Let $\gamma\neq 1/2$ and $\mu/\sigma^2\neq 1/2$. $\phi, c, s$ is a solution to the free boundary problem \eqref{new fbvp1x}--\eqref{new fbvp4x} if and only if $s$ and c satisfy the following system of non-linear equations:
\begin{align}\label{eq: its a boy}
\left(\frac{c+(1-\varepsilon)s}{c+1}\right)^{1-2\gamma}-1&=\frac{1-2\gamma}{1-2\gamma\pi^*}\frac{s^{1-2\gamma\pi^*}-1}{c+1},\\\label{eq: baby is born}
(1-\varepsilon) ^\frac{1}{2\gamma}s^{\pi^*}&=\frac{c+(1-\varepsilon)s}{c+1}.
\end{align}
\end{lemma}

\begin{proof}
The initial value problem \eqref{new fbvp1x}, \eqref{new fbvp3x} and \eqref{new fbvp5x}, parameterized in $c$, has the explicit solution \eqref{eq: sol in c}. What remains is to involve the boundary conditions
\eqref{new fbvp2x} and \eqref{new fbvp4x}. Starting from \eqref{eq: precessor} and using \eqref{new fbvp2x} yields
\begin{equation}
\frac{(c+(1-\varepsilon)s)^{1-2\gamma}}{1-2\gamma}=\frac{(s^{1-2\gamma\pi^*}-1)}{(1-2\gamma\pi^*)(c+1)^{2\gamma}}+\frac{(c+1)^{1-2\gamma}}{1-2\gamma},
\end{equation}
from which \eqref{eq: its a boy} follows.
Using \eqref{eq: grandpa}, \eqref{new fbvp2x} and \eqref{new fbvp4x} one obtains
\begin{equation*}
\frac{(1-\varepsilon)}{(c+(1-\varepsilon)s)^{2\gamma}}=\frac{s^{-2\gamma\pi^*}}{(c+1)^{2\gamma}}.
\end{equation*}
Taking the $2\gamma$ths root, one obtains \eqref{eq: baby is born}. The proof of the converse implication is similar.
\end{proof}
\section{Asymptotics of the free boundaries}
Recall that $\pi^*=\frac{\mu}{\gamma \sigma^2}$ and note that
\begin{equation}\label{eq: rueck trafo}
\zeta_-=\frac{1}{c},\quad\zeta_+=\frac{s}{c}
\end{equation}
and the associated trading boundaries $\pi_\pm$ satisfy
\begin{equation}\label{eq: rueckendeckung}
\pi_\pm :=\frac{\zeta_\pm}{1+\zeta_\pm}.
\end{equation}

We introduce the abbreviations
\[
\bar c:=\frac{1-\pi^*}{\pi^*},\quad\Delta:=\left(\frac{6}{\gamma\pi^*(1-\pi^*)}\right)^{1/3}\varepsilon^{1/3}.
\]
\begin{proposition}\label{th free}
For sufficiently small $\varepsilon>0$ the free boundary problem \eqref{new fbvp1x}--\eqref{new fbvp4x} has a unique solution $(h(\zeta),c,s)$. Moreover, 
 the following asymptotics hold as $\varepsilon\rightarrow 0$:
\begin{align}\label{c asymptotics}
c&=\bar c+\frac{1-\pi^*}{2\pi^*}\Delta+\frac{(1-\pi^*)(3-\pi^* (2\gamma+1))}{12\pi^*}\Delta^2\\\nonumber&\quad\quad\quad\quad-\frac{(\pi^* -1) \left(\left(4 \gamma ^2+22 \gamma +1\right) (\pi^*) ^2-24 (2 \gamma +1) \pi^* +36\right)}{360 \pi^* } \Delta^3+O(\varepsilon^{4/3}),\\\label{s asymptotics}
s&=1+ \Delta+\frac{\Delta^2}{2} + \frac{1}{180} \left(\left(4 \gamma ^2-8 \gamma +1\right) (\pi^*) ^2+3 (4 \gamma -3) \pi^* +36\right)\Delta^3\\\nonumber&\quad\quad\quad\quad +\frac{\left(8 \gamma ^2-26 \gamma +2\right) (\pi^*) ^2+2 (17 \gamma -9) \pi^* +27}{360}\Delta^4 +O(\varepsilon^{5/3}).
\end{align}
\end{proposition}
\begin{proof}
The proof is inspired by \cite[Proposition 6.1]{gerhold2013dual}, where a similar result is developed for log-utility from consumption and for unlevered strategies.\footnote{Similar methods to derive asymptotic expansions in small transaction costs are found in the papers \cite{gerhold.al.11,gerhold2012asymptotics,gm2020,gm2023}.} Having already solved the initial value problem \eqref{new fbvp1x}, \eqref{new fbvp3x} and \eqref{new fbvp5x}, parameterized in $c$, which has the explicit solution \eqref{eq: sol in c} it remains to involve the boundary conditions
\eqref{new fbvp2x} and \eqref{new fbvp4x}. A na\"ive approach would be to define
for sufficiently small $\delta$, the map $F:=(F_1,F_2)^\top$, where
\begin{align}\label{eq: sys1}
F_1(\delta,c,s)=&\left(\frac{c+(1-\delta^3)s}{c+1}\right)^{1-2\gamma}-1-\frac{1-2\gamma}{1-2\gamma\pi^*}\frac{s^{1-2\gamma\pi^*}-1}{c+1},\\\label{eq: sys2}
F_2(\delta,c,s)=&(1-\delta^3) ^\frac{1}{2\gamma}s^{\pi^*}-\frac{c+(1-\delta^3)s}{c+1},
\end{align}
and to show, by means of the implicit function theorem, that $F$ has a unique zero $(s(\delta), c(\delta))$ at $(c=\bar c, s=1)$ which is analytic in $\delta$. Note however that the implicit function theorem cannot be applied in this case: even though
$F(\delta_0=0, c_0=\bar c, s_0=s)=0$, the Jacobian vanishes at the critical point $(0,\bar c,1)$.

Consider the levered case only, as the other case can be proved quite similarly. In this case $s<1$. Having a look at eq. \eqref{eq: grandpa}, one sees that for $z<1$, $z$ sufficiently close to $z=1$, $\phi'(z)>0$ and since $\phi(1)=1$ this implies $\phi(z)<1$. Since $\phi=\phi(z,c)$ is an analytic function in $(z,c)$ near $c=\bar c$ and $z=1$, it satisfies an expansion of the form
\[
\phi(z,c)=1+(z-1)+\sum_{i\geq 2}\sum_{j\geq 0}a_{ij}(z-1)^i(c-\bar c)^j
\]
with coefficients $a_{ij}$ which can be calculated recursively. Furthermore, $a_{0j}=a_{1j}=\delta_{0j}$ for $j\geq 0$ due to
the initial conditions \eqref{new fbvp3x} \& \eqref{new fbvp5x}. One now solves for $c,s$ invoking the terminal conditions
\eqref{new fbvp2x} \& \eqref{new fbvp4x}. The latter imply that
\[
\varepsilon s=s-\phi(s,c),\quad\textrm{and}\quad \phi(s,c)-s \phi'(z=s,c)=0.
\]
Dividing by $s-1$, reflecting that the solution $s=1$ is not interesting, a Taylor expansion
yields
\begin{equation}\label{2eqs}
\frac{\phi(s,c)-s  \phi'(s,c)}{s-1}=\sum_{i\geq 0}\sum_{j\geq 0}b_{ij}(s-1)^i(c-\bar c)^j
\end{equation}
for certain coefficients $b_{ij}$. By using  \eqref{new fbvp1x}, \eqref{new fbvp3x} \& \eqref{new fbvp5x}
and L'Hospital's rule, one gets
\[
b_{0,0}=\lim_{z\rightarrow 1, c\rightarrow \bar c}\frac{\phi(z,c)-z  \phi'(z,c)}{z-1}=-\phi''(1,\bar c)=0,
\]
and, further by a twofold application of L'Hospital's rule,
\[
b_{1,0}=\lim_{z\rightarrow 1, c\rightarrow \bar c}\frac{\phi(z,c)-z  \phi'(z,c)}{(z-1)^2}=-\lim_{z\rightarrow 1, c\rightarrow \bar c}\phi^{(3)}(z,c)=2\gamma\pi^*(1-\pi^*)\neq 0.
\]
Hence the implicit function theorem is applicable and yields $s(c)=H(c)$ as a function of $c$ such that
\[
H(\bar c)=1,\quad H'(\bar c)=\frac{2\pi^*}{1-\pi^*}.
\]
Inserting this function into \eqref{eq: baby is born} one obtains
the problem
\[
g(c,\delta):=(1-\delta^3)^{\frac{1}{2\gamma}}H^{\pi^*}(c)-\frac{c+(1-\varepsilon) H(c)}{c+1}=0.
\]
Since $g(c=\bar c, \delta=0)=0$ and $\partial_c g(c,\delta)=\frac{1}{\pi^*}\neq 0$, one can apply the implicit
function theorem which asserts that for sufficiently small $\delta$, a unique and analytic solution $c=c(\delta)$
exists to $g(c,\delta)=0$ and $c(0)=\bar c$. Therefore,  $c(\delta), s(\delta)=H(c(\delta))$ is the unique solution 
of our problem for small $\delta$.

Finally, one derives the asymptotic formulas \eqref{c asymptotics}, \eqref{s asymptotics}: Let $(\phi,c,s)$ be the unique solution of \eqref{new fbvp1x}--\eqref{new fbvp4x}.  Due to Lemma \ref{nec s c}, $s$ and $c$ satisfy the system \eqref{eq: its a boy}--\eqref{eq: baby is born}. Substitute $c=c(s)$ from \eqref{eq: baby is born} into \eqref{eq: its a boy}, and replace $\varepsilon$ by $\delta^3$
in all equations. Then one plugs into the modified equation \eqref{eq: its a boy} a power series Ansatz for $s$, namely $s=1+s_1 \delta+\dots s_6 \delta^6$. Developing both sides as power series in $\delta$ and comparing coefficients leads to \eqref{s asymptotics}. This result is then plugged into \eqref{eq: baby is born}, yielding quite similarly, \eqref{c asymptotics}.
\end{proof}
\begin{remark}\label{rem: free}
Using the formulae \eqref{eq: rueck trafo}, the asymptotics \eqref{eq: asymptotics zeta shadow} for the trading boundaries $\widetilde\zeta_\pm$ in terms of the risky-safe ratio follow from the asymptotics of Proposition \ref{th free}. The asymptotics \eqref{eq: trading boundaries shadow} for $\widetilde \pi_\pm$ then follows from the relationship \eqref{eq: rueckendeckung}.
\end{remark}

\end{document}